%% file: root.tex
\documentclass[letterpaper, 10 pt, conference]{ieeeconf}  %

\IEEEoverridecommandlockouts                              %
\usepackage{xcolor}
\usepackage{amsmath,amsfonts,amssymb}
\usepackage{graphicx}
\usepackage{optidef}
\usepackage{dsfont}

\usepackage{amsmath, amssymb, amsthm}
\usepackage{algorithm}
\usepackage{algorithmic}
\usepackage{url}

\newtheorem{assumption}{Assumption}
\newtheorem{remark}{Remark}
\newtheorem{proposition}{Proposition}
\newtheorem{lemma}{Lemma}

\makeatletter
\renewenvironment{proof}[1][\relax]{\par
  \pushQED{\qed}%
  \normalfont \topsep6\p@\@plus6\p@\relax
  \trivlist
  \item[\hskip\labelsep\itshape
    \ifx#1\relax \proofname\else\proofname{} of #1\fi\@addpunct{.}]\ignorespaces
}{%
  \popQED\endtrivlist\@endpefalse
}
\makeatother

\input{commands}

\title{\LARGE \bf
Deceptive Planning Exploiting Inattention Blindness
}

\author{Mustafa O. Karabag$^{1}$, Jesse Milzman$^{2}$, Ufuk Topcu$^{1}$ 
\thanks{$^{1}$M. O. Karabag and U. Topcu are with the University of Texas at Austin, Austin, TX 78712 {\tt\small \{karabag, utopcu\}@utexas.edu}.} \thanks{$^{2}$J. Milzman is with the U.S. Army Research Laboratory, Adelphi, MD 20783 {\tt\small jesse.m.milzman.civ@army.mil}.} %
}
\usepackage{tikz}
\usepackage{pgfplots}
\pgfplotsset{compat=1.18} %

\usetikzlibrary{automata,positioning}
\usepackage{hyperref} %

\usepackage{subcaption}
\begin{document}

\maketitle
\thispagestyle{empty}
\pagestyle{empty}

\begin{abstract} 
We study decision-making with \textit{rational inattention} in settings where agents have perception constraints. In such settings, inaccurate prior beliefs or models of others may lead to \textit{inattention blindness}, where an agent is unaware of its incorrect beliefs. We model this phenomenon in two-player zero-sum stochastic games, where Player 1 has perception constraints and Player 2 deceptively deviates from its security policy presumed by Player 1 to gain an advantage. We formulate the perception constraints as an online sensor selection problem, develop a value-weighted objective function for sensor selection capturing rational inattention, and propose the greedy algorithm for selection under this monotone objective function. When Player 2 does not deviate from the presumed policy, this objective function provides an upper bound on the expected value loss compared to the security value where Player 1 has perfect information of the state. We then propose a myopic decision-making algorithm for Player 2 to exploit Player 1's beliefs by deviating from the presumed policy and, thereby, improve upon the security value. Numerical examples illustrate how Player 1 persistently chooses sensors that are consistent with its priors, allowing Player 2 to systematically exploit its inattention.
\end{abstract}

\input{sections/Introduction}
\input{sections/Preliminaries}

\input{sections/AttentionBasedOnlinePerception}
\input{sections/DeceptivePlanningForAttention}

\input{sections/Experiments}

\input{sections/Conclusions}
\bibliographystyle{IEEEtran}
\bibliography{refs}

\input{appendix2}

\end{document}

%% file: commands.tex
\newcommand{\simplex}{\Delta}
\newcommand{\genericRV}{X}
\newcommand{\genericSupport}{\mathcal{X}}
\newcommand{\altgenericRV}{Y}
\newcommand{\altgenericSupport}{\mathcal{Y}}
\newcommand{\entropy}[1]{H(#1)}
\newcommand{\probability}[1]{p(#1)}
\newcommand{\mdp}{\mathcal{M}}
\newcommand{\states}{S}
\newcommand{\actions}{A}
\newcommand{\transition}{P}
\newcommand{\discount}{\gamma}
\newcommand{\reward}{r}
\newcommand{\state}{s}
\newcommand{\action}{a}
\newcommand{\altstate}{q}

\newcommand{\game}{\mathcal{G}}
\newcommand{\initialstate}{\state_{0}}
\newcommand{\policy}{\pi}
\newcommand{\gamevalue}{V}
\newcommand{\qvalue}{Q}
\newcommand{\distribution}{d}
\newcommand{\observationset}{\Omega}
\newcommand{\observation}{\omega}
\newcommand{\observationfunction}{O}
\newcommand{\indexset}{I}
\newcommand{\belief}{b}
\newcommand{\observationcost}{c}
\newcommand{\costbudget}{C}
\newcommand{\randomvar}[1]{\pmb{#1}}

%% file: sections/Introduction.tex
\section{Introduction} \label{sec:introduction}

\textit{Rational inattention}~\cite{sims2003implications} is an economics model where agents make decisions with incomplete information, as acquiring or processing information is costly, or because the missing information does not add value to the agent's decisions. Decisions of such an agent rely on beliefs about its environment and other agents that share the same environment. Therefore, the agent must perform perception actions to obtain observations about these unknowns.

The agent updates its beliefs using observations. On the other hand, the accuracy of updates relies on the accuracy of the prior beliefs as well as the accuracy of the agent's models of the others. If these priors or assumptions are not accurate, the agent may suffer from \textit{inattention blindness}~\cite{mack1998inattentional}: The agent is not only incorrect about its beliefs but also it is unaware of this incorrectness since it does not collect observations to falsify these beliefs, and the received observations conform with the incorrect beliefs~\cite{klayman1995varieties}. In this case, the agent's adversaries can deviate from the presumed behavior to gain an advantage, exploiting inattention blindness. Such deceptive actions naturally emerge in different domains: in sports, a player makes a fake run to draw the attention of the opponent while another player who is presumed to be stationary and is not in the field of vision, makes an unnoticed run in the opposite direction; in military operations, a force repeatedly deploys decoy attack signals to cause the enemy not process these signals and then perform the attack unobserved; in cybersecurity, an attacker leaking data uses more primitive, low-bandwidth channels as these channels are not observed since the defender assumes that these channels would be highly inefficient for the attacker.

We model such interactions in two-player discounted zero-sum stochastic games. Player 1 does not fully observe the state; instead, it performs online perception at each step to choose sensors that refine its belief and decides on an action. Player 2 follows a known fixed policy, but its actions are not observable. The environment is a partially observable Markov decision process (partially observable MDP) from the perspective of Player $1$. 

To model the rational inattention for online perception, we propose an online sensor selection algorithm that aims to resolve the ambiguity about the states where Player 1's decisions change its value. In detail, for each state, we compute the conditional binary entropy of the state indicator variable given the selected sensors. We weight this conditional entropy by the gap between the highest and lowest $Q$-values. Summing these terms over all states yields our value-weighted entropy objective, which favors sensors that reduce the uncertainty about the high-stakes states where action choices lead to large value differences. Since this objective function is monotone in the chosen set of sensors, we propose using the greedy algorithm for online sensor selection. We show that, combined with the $Q_{\text{MDP}}$ heuristic~\cite{littman1995learning} and assuming that the player's belief matches the actual state distribution, this objective function provides a bound on the expected value loss for Player 1 compared to the case where it gets perfect observations of the state (i.e., compared to the optimal value of the MDP).

The value loss bound for Player $1$ holds in the zero-sum stochastic game setting if Player $2$ follows the presumed policy. To model deceptive planning exploiting inattention blindness, we consider that Player 2 deviates from this presumed policy. We model Player $2$ as choosing myopic deviations: given the belief of Player $1$, the minimizer Player $2$ chooses the action distribution with the lowest expected $Q$-value. We show that such deviations can only improve the expected return of Player $2$ since its expected discounted return for every time step is better than the security value.

We demonstrate this framework in two different numerical examples. In the first example, a defender protects a line, and the attacker aims to intrude at the farthest point from the defender. The proposed sensor selection approach results in the defender sensing the vertical position of the attacker, thereby making the attacker's horizontal moves unnoticed to gain an advantage. In the second numerical example, we quantitatively demonstrate the proposed framework in randomly generated games, highlighting that a player can exploit the inattentional blindness of the other to gain an advantage compared to the security value under different sensor selection methods.

\textit{Related work:}
In Economics, rational inattention models near-optimal decision-making with deliberately ignoring some information resources~\cite{sims2003implications,mackowiak2023rational}. For dynamic decision making, \cite{hebert2017rational} models rational inattention in a sequential information sampling problem where the decision-maker makes continuous-valued decisions to resolve state uncertainty that are subject to a cost constraint. 
For a partially observable MDP (POMDP), \cite{shafieepoorfard2016rationally} and \cite{shafieepoorfard2017rationally} model rational inattention as the co-design of the observation function and the control policy for a POMDP subject to a mutual information constraint between state and observations. We model the rational inattention in MDPs as an online sensor selection problem where sensors are chosen to resolve a value-weighted state uncertainty function.

Active perception aims to minimize belief uncertainty for the state to improve the accuracy of action decisions. Incorporating belief-dependent rewards in a POMDP implicitly encourages actions ~\cite {araya2010pomdp,spaan2015decision}. When the perception actions, i.e., sensor selections, are decoupled from the dynamics actions, the active perception problem can be modeled as an online sensor selection problem where the decision-maker chooses a subset of sensors from a set at each time step~\cite{satsangi2018exploiting,ghasemi2019online,krause2007near}. Existing works utilized entropy reduction for the state belief as the online sensor selection objective function~\cite{ghasemi2019online,satsangi2018exploiting}. While this approach offers desirable computational properties (e.g., submodularity), it may result in selecting sensors (paying attention to observations) that reduce state uncertainty but do not impact the expected return. Utilizing the value of decisions in the perfect information setting, we propose a state value-weighted entropy function that encourages the selection of the sensors that change the value and gives an upper bound on the expected value loss compared to the perfect information setting.

Deceptive planning aims to find a controller for an agent that exploits the lack of information or inaccurate beliefs of other agents~\cite{masters2017deceptive,karabag2021deception,rostobaya2023deception,suttle2025value}. Existing deceptive planning literature focuses on hiding targets from an observer by deviating from a behavioral model used by the observer to predict the movements of the ego agent~\cite{masters2017deceptive,10.5555/3463952.3464050,lewis2023deceptive,suttle2025value,chen2024deceptive}. Alternatively, deceptive motions that generate ambiguity can emerge as an equilibrium behavior in games~\cite{rostobaya2023deception}. We use a zero-sum game between two players and, similar to previous works, assume a rational behavior for the deceiving party~\cite{masters2017deceptive,lewis2023deceptive,suttle2025value,chen2024deceptive,rostobaya2023deception}. Unlike the existing works that often focus on the lack of information regarding targets, we focus on the lack of information regarding the game state and exploit partial observations. The works \cite{karabag2022exploiting,ma2024covert,fu2022almost} also focus on deception exploiting partial observations. These works focus on minimizing the detectability of a deviating single agent for fixed sensors, while we focus on a game between two players where sensors are chosen online.

%% file: sections/Preliminaries.tex
\section{Preliminaries and Notation}
We denote the $N$-dimensional probability simplex by $\simplex^{N}$, and the probability simplex over the set $C$ by $\simplex^{C}$. For random variables $X$ and $Y$, with a slight abuse of notation, we use $p(x)$, $p(x,y)$, $p(x|y)$ to denote the probability of $x$, joint probability of $x$ and $y$, and the conditional probability of $x$ given $y$. We denote the indicator function of a variable $x$ with $\mathds{1}_{y}(x)$, which equals $1$ if $x=y$ and $0$ otherwise. The random variable $\mathds{1}_{x}(X)$ is $1$ if $X = x$ and $0$ otherwise.

\subsection{Information Theoretical Quantities}
The entropy of a random variable $\genericRV$ with support $\genericSupport$ is 
\[
\entropy{\genericRV} =  - \sum_{x \in \genericSupport} \probability{x} \log_{2}  \probability{x}.
\]

The conditional entropy of a random variable $\genericRV$ given the random variable $\altgenericRV$ with support $\altgenericSupport$ is
\[
\entropy{\genericRV| \altgenericRV} =  - \sum_{x \in \genericSupport } \sum_{ y \in \altgenericSupport} \probability{x,y} \log_{2}  \frac{\probability{x,y}}{ \probability{y}}.
\]

\subsection{Markov decision processes and two-player zero-sum stochastic games}

A Markov decision process (MDP) $\mdp= (\states, \actions, \transition, \reward, \initialstate, \discount)$ is a tuple where $\states$ is a finite set of states, $\actions$ is a finite set of actions, $\transition:\states \times \actions \times \states \to [0,1]$ is the transition probability function such that $\sum_{\altstate \in \states}\transition(\state, \action, \altstate) = 1$ for all $\state \in \states, \action \in \actions$, $\reward: \states \times \actions \to [-R_{max}, R_{max}]$ is the reward function, and $\discount \in [0, 1)$ is the discount factor. A stationary policy $\policy:\states \times \actions \to [0,1]$ maps each state to an action distribution such that $\sum_{\action \in \actions} \policy(\state, \action) = 1$ for all $\state \in \states$. Under policy $\policy$, the expected discounted return from initial state $\state$ is \[\gamevalue^{\pi}(\state) = 
\mathbb{E}\left[\sum_{t=0}^{\infty} \discount^{t} \reward(\state_{t}, \action_{t}) \Bigg| \state, \policy \right],
\] where $\state_{0}\action_{0}\state_{1}\action_{1}\ldots$ is the random sequence of states and actions. We denote 
\[
\qvalue^{\policy}(\state, \action) =  
\mathbb{E}\left[\sum_{t=0}^{\infty} \discount^{t} \reward(\state_{t}, \action_{t}) \Bigg| \state_{0} = \state, \action^{1}_{0} = \action \right]
\] where actions $\action_{1}\action_{2}\ldots$ are sampled according to $\policy$.

There exists a stationary policy $\pi^{*}$ such that for all states $\state \in \states$, \[\gamevalue^{\pi^{*}}(\state) = \max_{\pi} \gamevalue^{\policy}(s).\] We use $\gamevalue^{*}(\state)$ to denote $\gamevalue^{\pi^{*}}(\state)$ and $\qvalue^{*}(\state, \action)$ to denote $\qvalue^{\policy^{*}}(\state, \action)$.

With a slight abuse of notation, we define a two-player zero-sum stochastic game $\game = (\states, \actions^{1}, \actions^{2}, \transition, \reward, \initialstate,  \discount)$ as a tuple where $\states$ is a finite set of states, $\actions^{1}$ is a finite set of actions for Player 1, $\actions^{2}$ is a finite set of actions for Player 2, $\transition:\states \times \actions^{1} \times \actions^{2} \times \states \to [0,1]$ is the transition probability function such that $\sum_{\altstate \in \states}\transition(\state, \action^{1}, \action^{2}, \altstate) = 1$ for all $\state \in \states, \action^{1}\in \actions^{1}, \action^{2} \in \actions^{2}$, $\reward: \states \times \actions^{1} \times \actions^{2} \to [-R_{max}, R_{max}]$ is the reward function for Player 1, $-\reward$ is the reward function for Player 2, and $\discount\in [0, 1)$ is the discount factor. A stationary policy $\policy^{i}:\states \times \actions^{i} \to [0,1]$ for player $i$ maps each state to an action distribution such that $\sum_{\action^{i} \in \actions^{i}} \policy^{i}(\state, \action^{i}) = 1$ for all $\state \in \states$.

Under policies $(\policy^{1}, \policy^{2})$, the discounted expected return of Player 1 from initial state $\state$ is
\[ \gamevalue^{\policy^{1}, \policy^{2}}(s) = 
\mathbb{E}\left[\sum_{t=0}^{\infty} \discount^{t} \reward(\state_{t}, \action^{1}_{t}, \action^{2}_{t}) \Bigg| \state_{0} = \state, \policy^{1}, \policy^{2}\right].
\]Player 1's goal is to maximize, and Player 2's goal is to minimize $\gamevalue^{\policy^{1}, \policy^{2}}(s)$. There exists an equilibrium pair of stationary policies $(\pi^{1, *}, \pi^{2, *})$ such that for all states $\state \in \states$, \[\gamevalue^{\pi^{1, *}, \pi^{2, *}}(\state) = \max_{\pi^{1}} \min_{\pi^{2}} \gamevalue^{\policy^{1}, \policy^{2}}(s) = \min_{\pi^{2}} \max_{\pi^{1}} \gamevalue^{\policy^{1}, \policy^{2}}(s),\] which is the security value for the players. We denote 
\[
\qvalue^{\policy^{1}, \policy^{2}}(\state, \distribution^{1}) =  
\mathbb{E}\left[\sum_{t=0}^{\infty} \discount^{t} \reward(\state_{t}, \action^{1}_{t}, \action^{2}_{t}) \Bigg| \state_{0} = \state, \action^{1}_{0} \sim \distribution^{1}\right]
\] where action $\action^{1}_{0}$ is drawn from $\distribution^{1}$, actions $\action^{1}_{1}\action^{1}_{2}\ldots$ are sampled according to $\policy^{1}$, and actions $\action^{2}_{0}\action^{2}_{1}\ldots$ are sampled according to $\policy^{2}$. Additionally, with an overload of notation, we denote
\begin{align*}
    &\qvalue^{\policy^{1}, \policy^{2}}(\state, \distribution^{1}, \distribution^{2}) =
    \\
&\mathbb{E}\left[\sum_{t=0}^{\infty} \discount^{t} \reward(\state_{t}, \action^{1}_{t}, \action^{2}_{t}) \Bigg| \state_{0} = \state, \action^{1}_{0} \sim \distribution^{1}, \action^{2}_{0} \sim \distribution^{2}\right]
\end{align*}
 where action $\action^{1}_{0}$ is drawn from $\distribution^{1}$, action $\action^{2}_{0}$ is drawn from $\distribution^{2}$, actions $\action^{1}_{1}\action^{1}_{2}\ldots$ are sampled according to $\policy^{1}$, and actions $\action^{2}_{1}\action^{2}_{1}\ldots$ are sampled according to $\policy^{2}$.

In the game setting, we use $\gamevalue^{*}(\state)$ to denote $\gamevalue^{\pi^{1, *}, \pi^{2, *}}(\state)$, $\qvalue^{*}(\state, \distribution^{1})$ to denote $\qvalue^{\pi^{1, *}, \pi^{2, *}}(\state, \distribution^{1})$, and $\qvalue^{*}(\state, \distribution^{1}, \distribution^{2})$ to denote $\qvalue^{\pi^{1, *}, \pi^{2, *}}(\state, \distribution^{1}, \distribution^{2})$.

\subsection{Partially observable MDPs and online sensor selection}
In a single-agent environment, consider an agent that does not have full observations of its own state. That is, the agent's environment is a \textit{partially observable MDP}, where the states, actions, and the transition probability function are defined the same as in an MDP. The agent collects a set of observations from sensors to maintain a belief $\belief$ over its state where $\belief:\states \to [0,1]$ and $\sum_{\state \in \states} \belief(s) = 1$. Let $\observationset^{1}, \ldots, \observationset^{N}$ be sets of observations associated with $N$ different sensors. Each set of observations is associated with an observation function $\observationfunction^{i}:\states \times \observationset^{i} \to [0,1]$ which maps state $s$ and observation $\observation$ to a probability value. Additionally, each sensor \(i\) has an associated cost $\observationcost^{i}$. As assumed in \cite{ghasemi2019online},
we also assume that the sensors are disjoint and independent given the state. 
\begin{assumption} \label{assumption:independence}
For all $i,j \in [N]$, $\observationset^{i} \cap \observationset^{j} = \emptyset$ and \[\probability{\observation^{i}, \observation^{j} | \state} = \probability{\observation^{i} | \state} \probability{\observation^{j} | \state} = \observationfunction^{i}(\state, \observation^{i}) \observationfunction^{j}(\state, \observation^{j})\] \text{ for all $\observation^{i} \in \observationset^{i}$, $\observation^{j} \in \observationset^{j}$, and $\state \in \states$}.  
\end{assumption}

Let $\belief_{t}$ denote the prior belief at time $t$ before observations and $\belief'_{t}$ denote the posterior belief after observations.
Given a set $\indexset_{t}$ of observation indices and a belief $\belief_{t}$, the agent updates its belief according to the Bayes' rule: 
\begin{equation}
    \belief'_{t}(\state) = \frac{\prod_{i \in \indexset_{t}} \observationfunction^{i}(\state, \observation^{i}) \belief_{t}(\state)}{\sum_{\altstate \in \states}\prod_{i \in \indexset_{t}} \observationfunction^{i}(\altstate, \observation^{i}) \belief_{t}(\altstate) } \label{eq:obsbeliefupdate}
\end{equation}

In the online perception setting, at each time $t$, the agent chooses a set $\indexset_{t}$ of sensors according to its current belief $\belief_{t}$. The updated belief is then used to make a decision.
Existing approaches often aim to reduce the uncertainty of the belief by minimizing the conditional entropy or variance of the belief. For example, given a belief $\belief_{t}$, the work \cite{ghasemi2019online} proposes to minimize \(\entropy{\randomvar{\state} | \cup_{i \in \indexset} \randomvar{\observation}^{i}}\) subject to $\sum_{i \in \indexset} \observationcost^{i} \leq \costbudget$ where the random state $\randomvar{\state}$ is distributed according to $\belief_{t}$ and proposes an approximate greedy algorithm to minimize this function.

\begin{remark}
    In the next sections, we consider agents with partial observations of the state. In order to provide notional simplicity, we do not formally define partially observable MDPs or partially observable stochastic games, nor do we focus on solution methods for these models, as our results rely solely on definitions from the fully observable settings.
\end{remark}

%% file: sections/AttentionBasedOnlinePerception.tex
\section{Value-loss Weighted Online Sensor Selection Modeling Rational Inattention}
\label{sec:blueplayer}
\usetikzlibrary{positioning, calc}

\textit{Rational inattention theory}~\cite{sims2003implications} proposes that a decision-making agent may prefer not to acquire or process certain information resources if the acquisition or processing of these resources is costly or if they do not affect the optimality of the agent's decisions for its objective. 

In this section, we propose a value-loss weighted online active sensor selection algorithm to capture rational inattention. While our goal is to develop a framework for deception exploiting inattention blindness in two-player games, for notational simplicity, in this section, we consider that the decision-making agent's environment is a partially observable MDP, i.e., only the considered agent takes actions, and the other player's policy is a known fixed policy. In Section \ref{sec:exploitation}, we study the setting where the other player deviates from the assumed policy and generalize the ideas in this section to two-player stochastic games. 

\subsection{Online Optimistic Sensor Selection} \label{sec:myopiconlineperception}
In a partially observable environment, reducing the state entropy uniformly may cause the agents to choose sensors that do not necessarily alter the optimal decision. For example, a driver slows down if the next car forward in the lane slows down, regardless of the other cars' speeds. However, reducing the state entropy may require measuring the other cars' speeds if they are encoded in the state space. 

\textbf{Modeling rational inattention in online perception:}
To capture the effect of value in online perception, we propose to minimize the objective function
\begin{equation} \label{eq:objfunction}
   \sum_{\state \in \states} \entropy{\mathds{1}_{\state}(\randomvar{\state}) | \cup_{i \in \indexset} \randomvar{\observation}^{i} } \Delta(\state)
\end{equation} where for state $\state$ \[\Delta(\state) = \max_{\action \in \actions} \qvalue^{*}(\state, \action) - \min_{\action \in \actions} \qvalue^{*}(\state, \action)\] represents the maximum value change for different actions. 

Minimizing \eqref{eq:objfunction} encourages the selection of sensors that resolve the uncertainty about the states where the decisions of the agent change significantly the value. The conditional binary entropy \(\entropy{\mathds{1}_{\state}(\randomvar{\state}) | \cup_{i \in \indexset} \randomvar{\observation}^{i} }\) measures the uncertainty about whether the agent is at \(\state\) or not given the observations. This term goes to \(0\) if \(\probability{\randomvar{s} = \state| \cup_{i \in \indexset} \randomvar{\observation}^{i}} \to 0\) or \(\probability{\randomvar{s} = \state| \cup_{i \in \indexset} \randomvar{\observation}^{i}} \to 1\). The weighting $\Delta(s)$ measures how much the decisions of the agent change the value. Consequently, the objective function \eqref{eq:objfunction} encourages the selection of sensors that reduce uncertainty for high-stakes states, and the agent rationally does not pay attention to less valuable sensors.  

We call the objective \textit{optimistic} since it focuses on resolving state uncertainty for at the current time step. The weighting term \(\Delta(\state)\) uses $Q$-value differences myopically, inherently assuming that the agent will achieve the optimal value under perfect information (i.e., the state is known) in the subsequent steps.

We note that the term \(\entropy{\mathds{1}_{\state}(\randomvar{\state}) | \cup_{i \in \indexset} \randomvar{\observation}^{i} }\) is a monotone non-increasing in the sensor set \(\indexset\) since conditioning does not increase entropy~\cite{cover1999elements} and, hence, for each additional sensor $j$, we have \(\entropy{\mathds{1}_{\state}(\randomvar{\state}) | \cup_{i \in \indexset} \randomvar{\observation}^{i} } \geq \entropy{\mathds{1}_{\state}(\randomvar{\state}) | \cup_{i \in \indexset \cup \lbrace j \rbrace} \randomvar{\observation}^{i} }\). 

Minimizing \eqref{eq:objfunction} is a combinatorial optimization problem, which is NP-hard in general~\cite{papadimitriou1998combinatorial}. Given this monotonicity property, we propose using a greedy algorithm. 

\begin{algorithm} 
\caption{Greedy Algorithm for Rational Inattention} \label{algo:rationalinattention}
\begin{algorithmic}[1]
\REQUIRE Initial belief \( b_{t} \), cost budget $C$, observation functions $O^{1}, \ldots, O^{N}$.
\STATE Index set \( \indexset_{t} \gets \emptyset \)
\WHILE{$\sum_{i \in \indexset_{t}} c^{i} \leq C$ and $\indexset_{t} \neq \lbrace 1, \ldots, N\rbrace$}
    \STATE Select index \( j \) such that:
    \[
    j = \arg\min_{j \in \lbrace 1, \ldots, N\rbrace \setminus \indexset_{t}}  \sum_{\state \in \states} \entropy{\mathds{1}_{\state}(\randomvar{\state}_{t}) \mid \cup_{i \in \indexset_{t} \cup \{j\}} \randomvar{\observation}^{i}} \Delta(\state)
    \]
    \STATE Update \( \indexset_{t} \gets \indexset_{t} \cup \{j\}. \)
\ENDWHILE
\RETURN Final index set \( \indexset_{t} \)
\end{algorithmic}
\end{algorithm}

Other potential stopping criteria include stopping if the objective value \eqref{eq:objfunction} is below a constant value. In the next section, we show that this criterion, combined with the $Q_{\text{MDP}}$ heuristic for action selection~\cite{littman1995learning}, guarantees a near-optimal value loss compared to the perfect information case.

\begin{remark}
We note that \eqref{eq:objfunction} is not necessarily submodular due to the indicator function. Following a similar approach to the proof of Lemma 1 in \cite{ghasemi2019online}, one can also show that \eqref{eq:objfunction} is submodular and the greedy selection algorithm is approximately optimal under the additional assumption that for all belief $\belief \in \simplex^{\states}$, $\state \in \states$, $\randomvar{\state} \sim \belief$, $\observation^{i}, \observation^{j} \in \observationset$ \[\probability{\observation^{i}, \observation^{j} | \mathds{1}_{\state}(\randomvar{\state} )} = \probability{\observation^{i} | \mathds{1}_{\state}(\randomvar{\state} )} \probability{\observation^{j} | \mathds{1}_{\state}(\randomvar{\state} )}.\]
\end{remark}

\usetikzlibrary{automata,positioning}

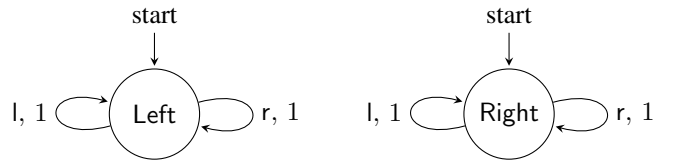
\begin{figure}[b]
    \centering
\begin{tikzpicture}[>=stealth,shorten >=1pt,auto,node distance=3.5cm,
                    state/.style={circle,draw,minimum size=1.2cm}]

  \node[state,initial,initial where=above] (up) {$\mathsf{Left}$};
  \node[state,initial,initial where=above, right=of up] (down) {$\mathsf{Right}$};

  \path (up) edge[loop left]  node {$\mathsf{l},\,1$} (up)
             edge[loop right] node {$\mathsf{r},\,1$} (up);

  \path (down) edge[loop left]  node {$\mathsf{l},\,1$} (down)
               edge[loop right] node {$\mathsf{r},\,1$} (down);

\end{tikzpicture}

    \caption{An MDP with two possible initial states. A label $a,p$ shows a transition that happens with probability $p$ under action $a$. The actions that match the state gives a reward of $1$ while the others give a reward of $0$, i.e., $\reward(\mathsf{Left}, \mathsf{l}) = \reward(\mathsf{Right},\mathsf{r}) = 1$ and $\reward(\mathsf{Left}, \mathsf{r}) = \reward(\mathsf{Right},\mathsf{l}) = 0$}
    \label{fig:beliefmismatch}
\end{figure}

\subsection{Value Loss Compared to Perfect Information} \label{sec:valueloss}

Consider that, after the sensor selection, the agent uses the updated belief $\belief_{t}'$ with the $Q_{\text{MDP}}$ heuristic~\cite{littman1995learning} to choose its action at time \(t\), i.e.,  
\begin{equation} \label{eq:qmdpactionselection}
   \action_{t} \in \arg \max_{\action \in \actions} \sum_{\state \in \states} \belief_{t}'(\state) \qvalue^{*}(\state, \action). 
\end{equation}

If the sensor selection guarantees that \eqref{eq:objfunction} is below a constant value for every time step, this action selection rule guarantees that the expected value loss of the agent is bounded compared to the perfect information case, i.e., the optimal initial state value of the MDP. 

\begin{proposition} \label{prop:expectedvalueloss}
    Let $v$ be the expected discounted return under the action selection rule \eqref{eq:qmdpactionselection} and the sensor sets $\indexset_{t}$ chosen in Algorithm \ref{algo:rationalinattention} satisfy \[\sum_{\state \in \states} \entropy{\mathds{1}_{\state}(\randomvar{\state}_{t}) | \cup_{i \in \indexset_{t}} \randomvar{\observation}^{i} } \Delta(\state) \leq \alpha \]
    for all $t\geq0$. Then,  
    \[\gamevalue^{*}(\state_{0}) - v \leq \frac{\alpha}{(1-\discount)}. \]
\end{proposition}

The proof is available in the appendix. The proof of Proposition $1$ relies on bounding the expected value loss for different states. If a state has very low binary entropy, then the state either has a very low belief probability or a very high belief probability. Since the value loss due to the current decision is bounded by the $Q$-value differences, the expected value loss is small due to the states with low belief probabilities. For states with high belief probabilities, the chosen action may already be optimal, resulting in no value loss. If the action is not optimal, then due to the $Q_{\text{MDP}}$ decision rule, it is guaranteed that the expected value from the state with high belief probability is bounded since the other states with bounded expected value losses dominate the decision.

We remark that the expected value loss is with respect to the agent's initial belief. If this belief does not match the initial state distribution, then the actual expected loss of the agent may not vanish as \(c \to 0\). As an example, consider the MDP given Fig. \ref{fig:beliefmismatch}. The agent's belief is $\belief(\mathsf{left}) = 1$ and $\belief(\mathsf{right}) = 0$ while the initial state is $\mathsf{right}$. Consider that there is a single sensor such that $\observationset = \lbrace \mathsf{null} \rbrace$ and $\observationfunction(\mathsf{right},\mathsf{null}) = \observationfunction(\mathsf{left},\mathsf{null}) = 1$. If $\initialstate  \sim \belief$, then $\sum_{\state \in \states} \entropy{\mathds{1}_{\state}(\randomvar{\state}_{0}) | \cup_{i \in \indexset} \randomvar{\observation}^{i} } \Delta(\state) = 0$ since the agent is certain that it is at state $\mathsf{left}$. However, the value loss with respect to the actual initial state distribution is $1/(1-\discount)$, which is the maximum value gap for the MDP in Fig. \ref{fig:beliefmismatch}. We note that this gap is due to the confirmation bias and occurs for partially observable MDPs in general if the initial belief is not accurate. 

%% file: sections/DeceptivePlanningForAttention.tex
\section{Deceptive Deviations to Exploit Inattention Blindness}

\label{sec:exploitation}
In Section \ref{sec:blueplayer}, we discussed rational inattention to model the perception decisions of a single agent acting alone in an MDP. We now focus on how an adversarial agent can exploit this perception method to gain an advantage in a zero-sum stochastic game. In the next sections that consider the zero-sum game, we refer to the agent with rational inattention as Player 1 and the adversarial agent as Player 2. 

Analogous to the $Q_{\text{MDP}}$ heuristic given in \eqref{eq:qmdpactionselection}, for a zero-sum stochastic game, we have the following action selection rule for Player 1:
\begin{equation} \label{eq:pl1distselection}
    \action^{1}_{t} \sim d^{1,*}_{t}= \arg \max_{d^{1} \in \mathcal{D}}   \sum_{\state \in \states} \belief_{t}'(\state) \qvalue^{*}(\state, d^{1}). 
\end{equation}
where $\mathcal{D} = \lbrace d | \exists \state , d = \policy^{1,*}(\state)\rbrace$, i.e., $\mathcal{D}$ is the set of action distributions utilized by Player 1 under the equilibrium policy. In words, given a set $\mathcal{D}$ of action distributions, Player 1 chooses the distribution that maximizes the expected return assuming that Player 2 follows the equilibrium policy, and it will have perfect observations in the next steps.

Analogous to Section \ref{sec:myopiconlineperception}, we define \[\Delta(\state) = \max_{d \in \mathcal{D}} \qvalue^{*}(\state, d) - \min_{d \in \mathcal{D}} \qvalue^{*}(\state, d)\] which represents the maximum value change for state \(\state\) for different action distributions.

Let $\belief'_{t}$ be the posterior belief after observations. In addition to the observations coming from the sensors, Player 1 updates its belief $\belief'_{t}$  using its action and Player 2's policy $\policy^{2,*}$ according to Bayes' rule: 
\begin{equation}
    \belief_{t+1}(\state') = \frac{ \sum_{\state \in \states} \sum_{\action^{2} \in \actions^{2}}\belief'_{t}(\state) \policy^{2,*}(\state, \action^{2}) \transition(\state, \action^{1}, \action^{2}, \state')}{\sum_{\altstate, \state \in \states}\sum_{\action^{2} \in \actions^{2}}\belief'_{t}(\state) \policy^{2,*}(\state, \action^{2}) \transition(\state, \action^{1}, \action^{2}, \altstate) } \label{eq:actbeliefupdate}
\end{equation}

\textbf{Confirmation Bias Leading to Inattention Blindness: }\textit{Inattentional blindness}~\cite{mack1998inattentional} is a psychological phenomenon in which individuals fail to recognize major and unexpected changes in their environment because they do not pay attention to the changes happening. In our framework, Player 1 may fail to realize that its belief is inaccurate because it may not perform enough perception, or if the received observations match the existing incorrect beliefs due to the confirmation bias.

Consider that Player 2 follows a fixed policy in the zero-sum game. In this case, the environment is an MDP from the perspective of Player 1. If Player 1 updates its belief according to this policy, then the performance guarantee given in Proposition \ref{prop:expectedvalueloss} holds for Player 1. However, if Player 2 employs a different policy, then Player 2 can gain an advantage. Player 1 may not even be aware of this deviation, i.e., have inattention blindness, resulting from misspecified priors and incorrect dynamics used in the belief updates. 

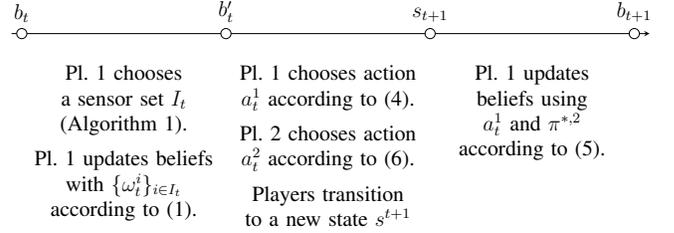
\begin{figure}[t]
\centering
\resizebox{\columnwidth}{!}{%
\begin{tikzpicture}[>=stealth, node distance=2.5cm, every node/.style={font=\large}]
    \draw[->] (-0.2,0) -- (12.3,0);

    \node[circle,draw,fill=white,inner sep=2pt,label=above: $b_{t}$] (b0) at (0,0) {};
    \node[circle,draw,fill=white,inner sep=2pt,label=above:$b_{t}'$] (b1) at (4,0) {};
    \node[circle,draw,fill=white,inner sep=2pt,label=above:$s_{t+1}$] (s) at (8,0) {};
    \node[circle,draw,fill=white,inner sep=2pt,label=above:$b_{t+1}$] (b2) at (12,0) {};

    \node[below=0.5cm of $(b0)!0.5!(b1)$,align=center,text width=3.6cm] 
        {Pl. 1 chooses a sensor set $\indexset_{t}$\\(Algorithm \ref{algo:rationalinattention}). \\\vspace{0.2cm} Pl. 1 updates beliefs with $\lbrace \observation^{i}_{t}\rbrace_{i \in \indexset_{t}}$ according to \eqref{eq:obsbeliefupdate}. };
    \node[below=0.5cm of $(b1)!0.5!(s)$,align=center,text width=3.5cm] 
        {Pl. 1 chooses action $\action^{1}_{t}$ according to \eqref{eq:pl1distselection}. \\ \vspace{0.2cm} Pl. 2 chooses action $\action^{2}_{t}$ according to \eqref{eq:pl2distselection}.
        \\
        \vspace{0.2cm}
        Players transition to a new state $\state^{t+1}$};
    \node[below=0.5cm of $(s)!0.5!(b2)$,align=center,text width=3.5cm] 
        {Pl. 1 updates beliefs using $\action^{1}_{t}$ and $\policy^{*,2}$ according to \eqref{eq:actbeliefupdate}.};
\end{tikzpicture}%
}
\caption{Sensor and action selection timeline for the players.}
\label{fig:decision-timeline}
\end{figure}

For example, consider the zero-sum game given in Fig. \ref{fig:beliefmismatchgame}. The security policy for Player 2 takes action $r$ with probability $1$ at start and the value of the game is $(1-\epsilon) \frac{\gamma}{1-\gamma}$. Consider that there are two sensors: the first one deterministically outputs $\mathsf{True}$ if the state is $\mathsf{LU}$ or $\mathsf{RU}$ and $\mathsf{False}$ otherwise, the second one deterministically outputs $\mathsf{True}$ if the state is $\mathsf{LU}$ and $\mathsf{LD}$ and $\mathsf{False}$ otherwise. Assuming that Player 2 took action $\mathsf{r}$ at the start, Player 1 needs to decide whether the state is $\mathsf{RU}$ or $\mathsf{RD}$, and given the observation from the first source, the belief entropy is $0$. On the other hand, the second sensor does not lower the state uncertainty. Instead, if Player $2$ takes action $\mathsf{l}$ at the start, inducing false beliefs, then the expected discounted return is $0$ since Player $1$ takes the other action, giving a reward of $0$. While Player $2$ deviates from the assumed policy, the observations that Player $1$ receives from sensor $1$ conform with the prior.

\begin{figure}[t]
\centering
\resizebox{\linewidth}{!}{%
\begin{tikzpicture}[>=stealth,shorten >=1pt,auto,
                    state/.style={circle,draw,minimum size=0.9cm,node distance=1.3cm}]

  \node[state] (l1) {$\mathsf{LU}$};
  \node[state,below=of l1] (l2) {$\mathsf{LD}$};

  \node[state,right=6cm of l1] (r1) {$\mathsf{RU}$};
  \node[state,below=of r1] (r2) {$\mathsf{RD}$};

  \path (l1) -- (l2) coordinate[midway] (lmid);
  \path (r1) -- (r2) coordinate[midway] (rmid);
  \node[state,initial where=above,initial] (start) at ($(lmid)!0.5!(rmid)$) {$\mathsf{Start}$};

  \path (start) edge[->] node[sloped,below=-2pt] {$\mathsf{a}, \mathsf{l}, 0.5$} (l1)
                edge[->,bend left=0]  node[sloped, above=-2pt] {$\mathsf{a}, \mathsf{l}, 0.5$} (l2);   %

  \path (start) edge[->] node[sloped,below=-2pt] {$\mathsf{a}, \mathsf{r}, 0.5$} (r1)
                edge[->,bend right=0] node[sloped, above=-2pt] {$\mathsf{a}, \mathsf{r}, 0.5$} (r2);

  \foreach \i in {1,...,2}{
    \path (l\i) edge[loop left]  node {$\mathsf{l}, \mathsf{a}, 1$} (l\i)
                edge[loop right] node {$\mathsf{r}, \mathsf{a}, 1$} (l\i);
  }

  \foreach \i in {1,...,2}{
    \path (r\i) edge[loop left]  node {$\mathsf{l}, \mathsf{a}, 1$} (r\i)
                edge[loop right] node {$\mathsf{r}, \mathsf{a}, 1$} (r\i);
  }

\end{tikzpicture}%
}

    \caption{A two-player stochastic game. A label $a^{1}, a^{2},p$ shows a transition that happens with probability $p$ under actions $a^{1}$ and $a^{2}$. For some $\epsilon \in (0,1)$, the rewards are \parbox{\linewidth}{\centering
$\reward(\mathsf{LU}, \mathsf{l}, \mathsf{a}) = \reward(\mathsf{LD}, \mathsf{r}, \mathsf{a}) = 1$,\\
$\reward(\mathsf{RU}, \mathsf{r}, \mathsf{a}) = \reward(\mathsf{RD}, \mathsf{l}, \mathsf{a}) = 1 - \epsilon$,\\
$0$ for others.}
}
\label{fig:beliefmismatchgame}
\end{figure}
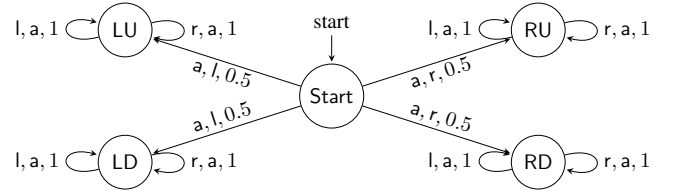

\subsection{Myopic Deceptive Planning to Exploit Incorrect Beliefs}

Consider that Player $2$ knows the observations received by Player $1$ and knows that Player $1$ assumes $\policy^{2,*}$ as Player $2$'s policy. Given this knowledge, Player $2$ can deviate from the equilibrium policy $\policy^{2,*}$ to gain an advantage. 

Given Player $1$'s belief $b^{*}_{t}$, Player $2$ follows the action selection rule:
\begin{subequations} \label{eq:pl2distselection}
    \begin{align}
    &\text{Compute $d^{1,*}_{t}$ according to \eqref{eq:pl1distselection}}.
    \\
        &d^{2,*}_{t} = \arg \min_{d^{2} \in \simplex^{\actions^{2}}} \sum_{\action^{1} \in \actions^{1}} \sum_{\action^{2} \in \actions^{2}} d^{1,*}_{t}(\action^{1})  d^{2}(\action^{2}) \nonumber
        \\
        &\quad \quad \quad \quad \quad \left(\reward(\state_{t}, \action^{1}, \action^{2}) + \mathbb{E}[\gamevalue^{*}(\state_{t+1}) | \action^{1}, \action^{2}]\right)
        \\
                &\action^{2}_{t} \sim  d^{2,*}_{t}
    \end{align}
\end{subequations}

Combining sensor and action selection mechanisms together, we have the timeline described in Fig. \ref{fig:decision-timeline}. Player $2$ maximizes the expected return assuming that the players will play the equilibrium policies in the following timesteps. Under these mechanisms, Player $2$'s deviations cannot decrease its expected return compared to the security value.

\begin{proposition} \label{prop:deceptiongain}
Let $\nu$ be the expected discounted return of Player $1$ under the perception and action decision rules defined in Fig. \ref{fig:decision-timeline}. Then, $\nu \leq \gamevalue^{*}(\initialstate)$.
\end{proposition}

The proof relies on the fact that at any time step, deviations of Player $2$ guarantee a value better than the security value for itself at the current state. Since the value does not get worse than the security value at any future time step, Player $2$'s expected return is better than the security value for the initial state. The complete proof is available in the appendix.

The decision-making rule described in \eqref{eq:pl1distselection} is myopic, as it maximizes the expected return assuming that Player $1$ will have perfect state information in the next time steps. In other words, \eqref{eq:pl2distselection} is a model predictive control method~\cite{rawlings2020model} using a decision window of $1$ with the security values as the terminal costs. One can extend the decision window to $T$ steps; however, this non-myopic approach has a computational complexity that exponentially grows with $T$ due to the possible realizations of states, actions, and observations for different timesteps. Therefore, the myopic approach is tractable. We remark that the myopic decision-making rule results in the security policy for the example given in Fig. \ref{fig:beliefmismatchgame}. A decision window of $2$ would result in Player $2$ taking action $\mathsf{l}$ at $\mathsf{Start}$.

%% file: sections/Experiments.tex
\section{Numerical Experiments}

We demonstrate the proposed deception model in two examples. The first example shows the behavior in a grid-world setting. The second example uses randomly generated games to quantitatively evaluate the performance of different online perception methods for Player 1 against different action selection methods of Player 2.

\subsection{Line  defense with coordinate sensors}
In this example, we consider two players in an $11\times 11$ grid-world shown in Fig. \ref{fig:trajs}. Let $(x^{i}_{t}, y^{i}_{t})$ represent the cell of player $i$ at time $t$. Player 1 starts from the blue cell $(6,1)$, and Player 2 starts from the red cell $(7,11)$. Player $1$ can only move horizontally on the top row, which implies $y^{1}_{t} = 11$ for all $t\geq 0$. Player 2 can move in all neighboring cells, including the diagonal ones. At each time step, a player moves to their target cell with probability $0.9$ and stays at its current cell with probability $0.1$. The game ends (i.e., the players transition to an absorbing state with no reward) after Player $2$ reaches the top row, i.e., $x^{2}_{t} = 11$. Player $1$'s goal is to capture Player $2$ at the top row. Player 2's goal is to reach the top row while having the maximum distance from Player $1$. Let $(x^{1}, y^{1}, x^{2}, y^{2})$ represent the state of the game. Formally, the reward is defined as $\reward((x^{1}, 11, x^{2}, 11), a^{1}, a^{2}) = -|x^{1} - x^{2}|$ and $\reward((x^{1}, 11, x^{2}, y^{2}), a^{1}, a^{2}) = 0$ for $y^{2} \neq 11$. The discount factor is $\discount = 0.99$. The value of the game is $\gamevalue^{*}(\initialstate) = -0.894$ when both players have perfect observations of the game state.

Player $1$ has two sensors available indexed with $1$ and $2$: $1)$ a sensor outputting Player $2$'s $x$ location $x^{2}_{t}$, $2)$ a sensor outputting Player 2's $y$ location $y^{2}_{t}$. Each of these two sensors outputs the true location with probability $0.7$ and adjacent locations in the same coordinate with probability $0.3$. The costs of these sensors are equal, $c^{1} = c^{2}$. Other than these two sensors, at all time steps, Player $1$ knows its own state, i.e., there exists a free sensor with index $3$ such that $\observationfunction^{3}(((x^{1}, y^{1}, x^{2}, y^{2})), (x^{1}, y^{1})) = 1$ and $\observationcost^{3} = 0$.

\input{sections/beliefcomparisonfig}

In addition to sensor $3$, Player $1$ can only use one of the sensors $1$ and $2$ due to the cost constraint since $C \leq c^{1} + c^{2}$. Player $2$ greedily chooses the additional sensor according to Algorithm \ref{algo:rationalinattention} and takes actions according to \eqref{eq:pl1distselection}.

In this setting, we consider two different policies for Player $2$. In the first case, Player 2 uses the equilibrium policy $\pi^{2,*}$ and in the second case Player $2$ uses \eqref{eq:pl2distselection} for action selection. We sample $10^3$ game runs for each of these cases. Fig. \ref{fig:trajs} shows $100$ of these game runs for Player $2$ in each case. In the first case, the estimated discounted return is $-0.958$. While Player $1$ lacks information, the return is only slightly worse than the value under perfect information. In the second case, the estimated discounted return is $-2.876$, indicating the gain for Player $2$. This gain aligns with the behavior observed in Fig. \ref{fig:trajs}. Under the equilibrium policy Player $2$ does not make horizontal moves in the earlier stages of the game since Player $1$ has time to cover these moves by moving in the same direction. Instead, Player $2$ makes random horizontal moves in the later stages of the game when Player $1$ does not have time. On the other hand, in the second case, Player $2$ deviates from the assumed policy and makes horizontal moves exploiting the beliefs of Player $1$. These early moves, which give Player $2$ an advantage, are unnoticed since Player $1$ expects Player $2$ not to move horizontally and therefore does not cover these moves. 

Fig. \ref{fig:heatmaps} shows this effect in more detail. Until time $t=6$, Player $2$ rarely activates the $x$ sensor. As a result, if Player $2$ makes an unexpected horizontal move, Player $1$ has an inaccurate belief. For example, at time steps $t=7,8,9,10$, we observe that while Player $2$ is near the edges, Player $1$ believes that Player $2$ is at the middle line. Furthermore, at time step $t=8$, we observe that even though the sensor provides the correct $x$ position of Player $2$ with high probability, these observations do not refine the belief because the behavior of Player $2$ is incorrectly modeled in the belief update. As a result, these accurate observations are treated as noisy signals of the presumed state and are effectively ignored. Player $1$ converges to more accurate beliefs over time with more usage of the $x$ sensor. However, the inattention blindness happening in the earlier stages results in delayed horizontal moves and loss of value for Player $1$.

\subsection{Randomly generated games}

In this example, we use 100 randomly generated games. Each game has 10 states, 4 actions for each player, 10 sensors, and 2 possible observations for each sensor. For each state \(\state\) and action pair \(\action^{1}, \action^{2}\), the transition probability distribution \([\transition(\state, \action^{1}, \action^{2}, \state_{1}), \ldots, \transition(\state, \action^{1}, \action^{2}, \state_{10})]\) is sampled from the 10-dimensional probability simplex uniformly randomly, and the reward \(\reward(\state, \action^{1}, \action^{2})\) is sampled from $[0,1]$ uniformly randomly. Similarly, for each state \(\state\) and sensor \(i\), the observation distribution \([\transition(\state, \observation^{i}_{1}), \transition(\state, \observation^{i}_{2})]\) is sampled from the 2-dimensional probability simplex uniformly randomly. The initial state is chosen uniformly randomly, and Player $1$ knows the initial state, i.e., its initial belief is a Dirac distribution. The discount factor is $0.9$.

Player 1 has the following sensor selection methods:
\begin{enumerate}
    \item Perfect information: Player \(1\) knows the state.
    \item Greedy weighted Bernouilli entropy: Use Algorithm \ref{algo:rationalinattention} to choose $k$ sensors.
    \item Greedy non-weighted entropy~\cite{ghasemi2019online}: Greedily minimizes \(\entropy{\randomvar{\state} | \cup_{i \in \indexset} \randomvar{\observation}^{i}}\) to choose $k$ sensors.
    \item Random: Uniformly randomly choose $k$ sensors.
    \item No observations: Player \(1\) only relies on its actions and the assumed policy of Player \(2\) for belief updates.
\end{enumerate}

For all methods, Player \(1\) uses the action selection rule \eqref{eq:pl1distselection}. For the first case, \eqref{eq:pl1distselection} generates the equilibrium policy $\policy^{1,*}$ since the equilibrium policy for Player 1 is optimal against the equilibrium policy for Player 2 at each state.  Also note that for the first case, \eqref{eq:pl2distselection} generates the equilibrium policy $\policy^{2,*}$ since the equilibrium policy for Player 2 is optimal against the equilibrium policy for Player 1 at each state. 

We use two different action selection rules for Player $1$ for each of the sensor selection methods:
\begin{enumerate}
    \item Equilibrium: Player $2$ follows  $\pi^{2,*}$.
    \item Belief exploitation: Player $2$ uses \eqref{eq:pl2distselection}.
\end{enumerate}

For each pair of sensor and action selection methods, we sample 100 game runs to estimate the discounted returns. In Fig. \ref{fig:returns}, we report the estimated discounted returns for Player 1. We observe that, as theoretically expected, perfect state information yields the highest return, and an increasing number of observations improves the returns for Player 1. We observe that the weighted Bernoulli entropy minimization (Algorithm \ref{algo:rationalinattention}) outperforms the non-weighted entropy minimization and random selection of sensors. We observe that regardless of the perception method, the action selection method based on beliefs, \eqref{eq:pl2distselection}, improves the returns for Player 2, matching the theoretical result given in Proposition \ref{prop:deceptiongain}.

\begin{figure}
    \centering
\begin{tikzpicture}

\definecolor{eqcolor}{RGB}{0, 0, 255}   %
\definecolor{becolor}{RGB}{255, 0, 0}    %
\begin{axis}[
    ybar,
    width=\linewidth,
    height=4.6cm,
    bar width=6pt,
    enlarge x limits=0.12,
    ylabel={Est. Player 1 Return},
    y label style={font=\footnotesize},
    symbolic x coords={
        Perfect State,
        Wt. (k=2),
        Non-wt. (k=2),
        Random (k=2),
        Wt. (k=1),
        Non-wt. (k=1),
        Random (k=1),
        No obs.
    },
    xtick=data,
    x tick label style={rotate=45, anchor=east, font=\footnotesize},
    y tick label style={font=\footnotesize},
    ymin=0,
    ymax=8,    
    tick pos=left,        
    ytick pos=left,     
    xtick pos=bottom,   
    legend style={at={(1,1)}, anchor=north east,legend columns=-1,    column sep=10pt,       draw=none, font=\small,},
    nodes near coords,
    nodes near coords style={font=\small, rotate=90, anchor=west},
        cycle list={{fill=eqcolor},{fill=becolor}}, 
]

\addplot coordinates {
    (Perfect State, 4.977)
    (Wt. (k=2), 4.698)
    (Non-wt. (k=2), 4.683)
    (Random (k=2), 4.659)
    (Wt. (k=1), 4.635)
    (Non-wt. (k=1), 4.633)
    (Random (k=1), 4.601)
    (No obs., 4.562)
};

\addplot coordinates {
    (Perfect State, nan)
    (Wt. (k=2), 3.052)
    (Non-wt. (k=2), 3.012)
    (Random (k=2), 2.960)
    (Wt. (k=1), 3.051)
    (Non-wt. (k=1), 2.946)
    (Random (k=1), 2.897)
    (No obs., 2.842)
};

\legend{Eq., Bl. Expl.}
\end{axis}
\end{tikzpicture}
\caption{Estimated discounted returns under different sensor selection methods for Player 1 and action selection methods for Player 2. Note that equilibrium and belief exploit policies are the same for Player 2 when Player 1 has the perfect state information.} 
    \label{fig:returns}
\end{figure}

%% file: sections/beliefcomparisonfig.tex
\begin{figure}[t]
        \centering
        \includegraphics[width=0.4\columnwidth]{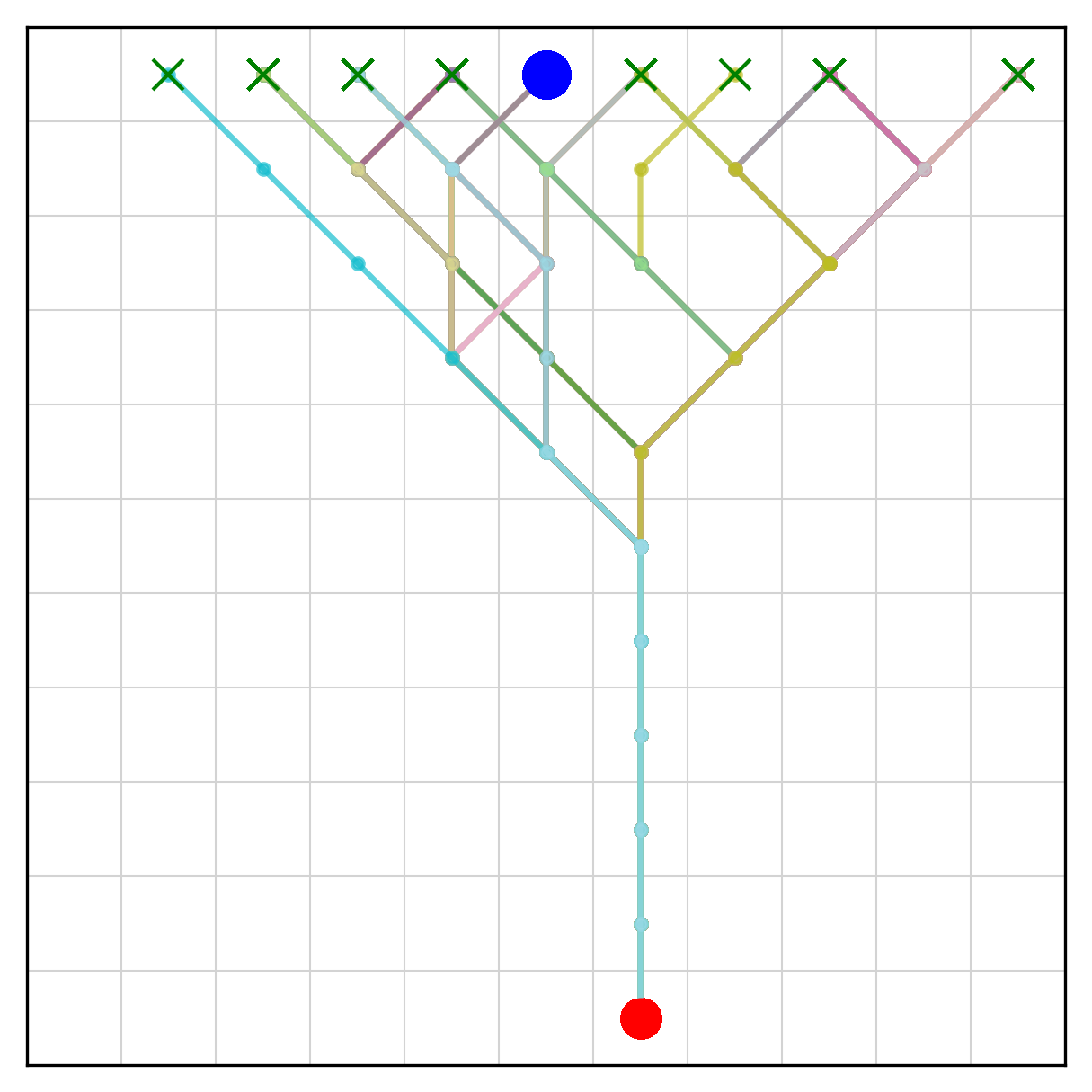}
\hfill
        \centering
        \includegraphics[width=0.4\columnwidth]{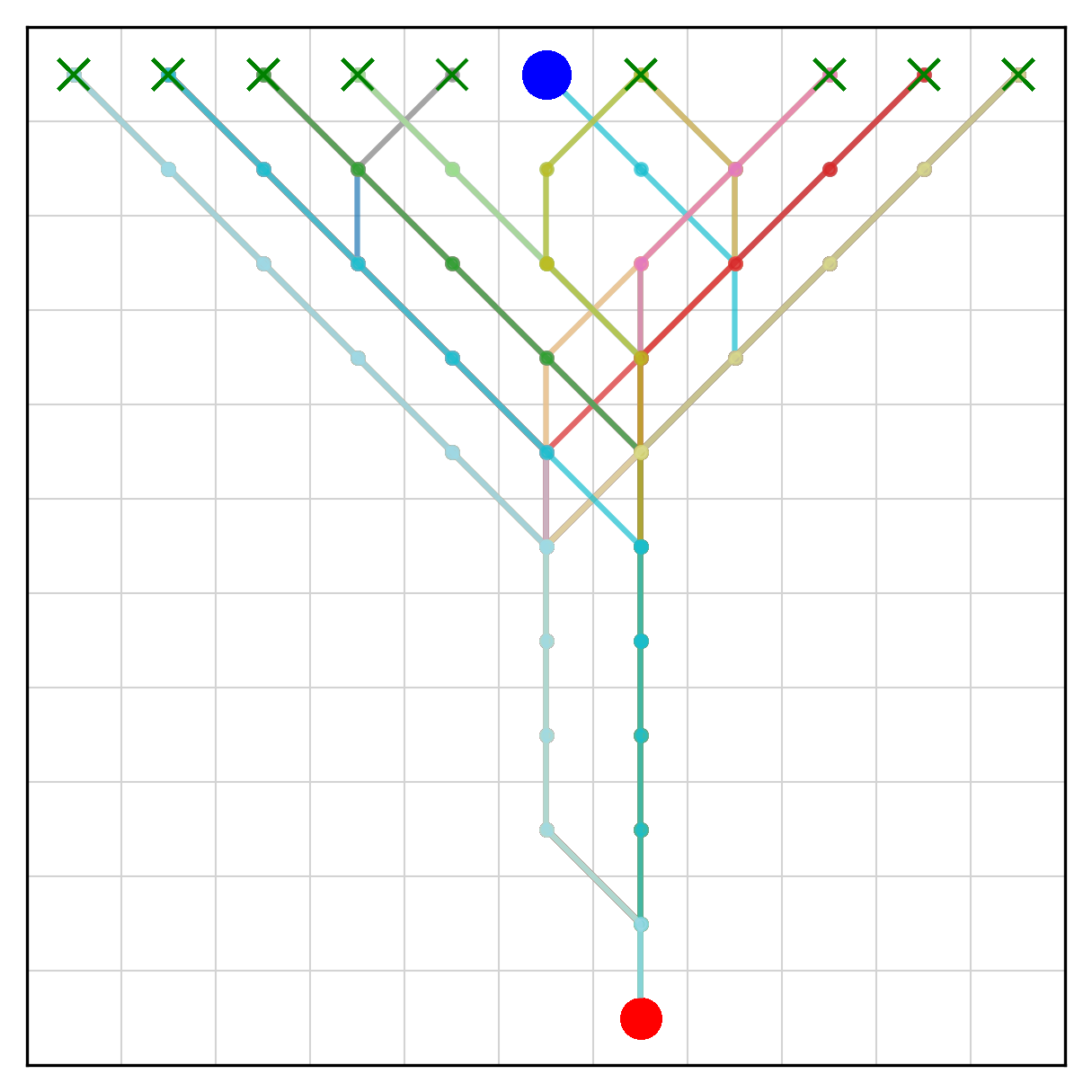}
\caption{(Left) Player 2 uses $\pi^{2,*}$, (Right) Player 2 uses \eqref{eq:pl2distselection} for action selection. 100 sample game runs for Player 2. Red dot indicates Player 1's start location, blue dot indicates Player 2's start location, and green crosses indicate the end.}
\label{fig:trajs}
\end{figure}

\newcommand{\heatmapTickSize}{\fontsize{16pt}{18pt}\selectfont}
\newcommand{\heatmapLabelSize}{\fontsize{24pt}{24pt}\selectfont}
\newcommand{\heatmapTitleSize}{\fontsize{30pt}{30pt}\selectfont}

\pgfplotsset{
  every axis/.append style={
    tick label style={font=\heatmapTickSize},
    label style={font=\heatmapLabelSize},
    title style={font=\heatmapTitleSize}
  }
}

\begin{figure*}[t]
\centering

\foreach \t in {5,6,7,8,9,10} {
    \begin{subfigure}{0.155\textwidth}
        \centering
        \resizebox{\textwidth}{!}{\input{figs/belief_distr_baseline/heatmap_t\t.tex}}
        \includegraphics[width=\textwidth, height=0.05\textwidth]{figs/belief_distr_baseline/heatmap_t\t_sensors.png}
        \vspace{0cm}

    \end{subfigure}
}
\foreach \t in {5,6,7,8,9,10} {
    \begin{subfigure}{0.155\textwidth}
        \centering
        \resizebox{\textwidth}{!}{\input{figs/belief_distr/heatmap_t\t.tex}}
        \includegraphics[width=\textwidth, height=0.05\textwidth]{figs/belief_distr/heatmap_t\t_sensors.png}
    \end{subfigure}
}

\caption{(Top) Player 2 uses $\pi^{2,*}$, (Bottom) Player 2 uses \eqref{eq:pl2distselection} for action selection. Conditional beliefs for the $x$ position (confusion matrices) and sensor choices (bar plots) of Player 1 at different time steps. Each (non-white) column of the heatmap is the average belief of Player 1 about Player 2's $x$ position conditioned on an actual $x$ position of Player 2. The intensity of the diagonal line shows the accuracy of the belief. The bar plots show the distribution of the chosen sensors, where red is the sensor for the $x$ position and blue is for the $y$ position. The demonstrated values are estimated using $10^3$ game runs.}
\label{fig:heatmaps}

\end{figure*}

%% file: sections/Conclusions.tex
\section{Conclusions}
We considered a deceptive planning framework based on rational inattention and inattention blindness, where two players interact in a zero-sum stochastic game. We proposed a rational inattention model for Player 1 for online perception, where Player 1 online chooses sensors of high value. We show that if Player 1 has accurate beliefs about the state, then this online perception method, combined with a simple action selection heuristic, results in a bounded loss compared to the case with perfect state information. Then, we considered an action selection method for Player 2 to deceive Player 1 by exploiting its beliefs. Deviations of Player 2 from the presumed policy by Player 1 lead to unnoticed incorrect beliefs for Player 1, leading to inattentional blindness. In future work, we aim to develop methods for Player $2$ that consider longer planning horizons to induce, maintain, and exploit inattention blindness.

%% file: appendix2.tex
\onecolumn

\section*{Appendix: Proofs}

\begin{lemma} \label{lemma:entropybound}
    Define the binary entropy function $h(p) = -p\log_{2}(p) - (1-p)\log_{2}(1-p)$. Then, $p \leq h(p)/2$ for all $p \in [0, 1/2]$.
\end{lemma}
\begin{proof}[Proposition \ref{prop:expectedvalueloss}]
    Note that $h(0) = 0$ and $h(1/2) = 1$. The inequality directly follows from these facts and Jensen's inequality using the concavity of $h(p)$ between $0$ and $1/2$: Consider a random variable $X$ taking value $0$ with probability $1-2p$ and $1/2$ with probability $2p$. Note that the expected value is $p$. Through this observation, we get \[h(\mathbb{E}[X]) = h(p) = h(0(1-2p) + 1/2(2p)) \geq \mathbb{E}[h(X)] = (1-2p) h(0)  + 2ph(1/2) = 2p.\]
\end{proof}

\begin{lemma} \label{lemma:stepvalueloss} Let $b$ be the initial belief, $b'$ be the updated belief after observing $\observation = \cup_{i \in \indexset} \observation^{i}$, and $\action(\observation)$   be the solution to $ \max_{\action \in \actions} \sum_{\state \in \states} \belief'(\state) \qvalue^{*}(\state, \action) $. 
 $$   \sum_{\state \in \states} \entropy{\mathds{1}_{\state}(\randomvar{\state}) |  \randomvar{\observation} } \Delta(\state) \geq \mathbb{E}_{ \randomvar{\observation}} \left[  \sum_{\state \in \states} \probability{\state | \observation} \left( \max_{\action \in \actions} \qvalue^{*}(\state, \action) - \qvalue^{*}(\state, \action( \observation))\right)\right]$$ where the randomness of $\randomvar{\observation} = \cup_{i \in \indexset} \randomvar{\observation}^{i}$ is over the randomness of $s \sim \belief$ and the randomness of the sensors in $\indexset$.
\end{lemma}
\begin{proof}
    Let $\belief'$ be the belief after observing $\observation$. We consider three cases. 
\begin{enumerate}
    \item $\belief'(\state) \leq 1/2$ for all $\state \in \states$. \label{case:allsmallprobs}
    \item $\belief'(\altstate) > 1/2$ for a single $\altstate \in \states$.
    \begin{enumerate}
        \item $\action(\observation) \in \arg \max_{\action \in \actions}  \sum_{\state \in \states}\belief'(\state)\qvalue^{*}(\state, \action)$. \label{case:likelydominates}
        \item $\action(\observation) \not\in \arg \max_{\action \in \actions}  \sum_{\state \in \states}\belief'(\state)\qvalue^{*}(\state, \action)$.  \label{case:likelyloses}
    \end{enumerate}
\end{enumerate}

\textit{Case \ref{case:allsmallprobs}:} Due to the definition of conditional entropy, Lemma \ref{lemma:entropybound}, and $\qvalue^{*}(\state, \action( \observation)) \geq \min_{\action \in \actions} \qvalue^{*}(\state, \action)$, we have
\begin{align}
    \sum_{\state \in \states} \entropy{\mathds{1}_{\state}(\randomvar{\state}) | \randomvar{\observation} } \Delta(\state) &= \sum_{\observation} \probability{\observation} \sum_{\state \in \states } \entropy{\mathds{1}_{\state}(\randomvar{\state}) | \observation}\Delta(\state)
    \\
    &\geq \sum_{\observation} \probability{\observation} \sum_{\state \in \states } 2 \probability{\state | \observation}\Delta(\state)
        \\
    &= \sum_{\observation} \probability{\observation} \sum_{\state \in \states } 2 \probability{\state | \observation} \left( \max_{\action \in \actions} \qvalue^{*}(\state, \action) - \min_{\action \in \actions} \qvalue^{*}(\state, \action)\right)
    \\
    &\geq \sum_{\observation} \probability{\observation} \sum_{\state \in \states } 2 \probability{\state | \observation} \left( \max_{\action \in \actions} \qvalue^{*}(\state, \action) - \qvalue^{*}(\state, \action( \observation))\right)
\end{align}
which shows the desired result.

\textit{Case \ref{case:likelydominates}:} Similarly, due to the definition of conditional entropy, Lemma \ref{lemma:entropybound}, and $\qvalue^{*}(\state, \action( \observation)) \geq \min_{\action \in \actions} \qvalue^{*}(\state, \action)$ and $\qvalue^{*}(\altstate, \action( \observation)) = \max_{\action \in \actions} \qvalue^{*}(\state, \action)$, we have
\begin{align}
    \sum_{\state \in \states} \entropy{\mathds{1}_{\state}(\randomvar{\state}) | \randomvar{\observation} } \Delta(\state) &= \sum_{\observation} \probability{\observation} \sum_{\state \in \states } \entropy{\mathds{1}_{\state}(\randomvar{\state}) | \observation}\Delta(\state)
        \\
    &= \sum_{\observation} \probability{\observation} \left( \sum_{\state \in \states \setminus \lbrace \altstate \rbrace } \entropy{\mathds{1}_{\state}(\randomvar{\state}) | \randomvar{\observation}} \Delta(\state) +  \entropy{\mathds{1}_{\altstate}(\randomvar{\state}) | \randomvar{\observation} } \Delta(\altstate) \right)
    \\
    &\geq \sum_{\observation} \probability{\observation} \left( \sum_{\state \in \states \setminus \lbrace \altstate \rbrace } 2 \probability{\state | \observation}\Delta(\state) +  \entropy{\mathds{1}_{\altstate}(\randomvar{\state}) | \randomvar{\observation} } \Delta(\altstate) \right)
        \\
    &\geq \sum_{\observation} \probability{\observation} \left( \sum_{\state \in \states \setminus \lbrace \altstate \rbrace } 2 \probability{\state | \observation}\left( \max_{\action \in \actions} \qvalue^{*}(\state, \action) - \qvalue^{*}(\state, \action( \observation))\right) +  \entropy{\mathds{1}_{\altstate}(\randomvar{\state}) | \randomvar{\observation} } \Delta(\altstate) \right)
    \\
    &\geq \sum_{\observation} \probability{\observation} \sum_{\state \in \states } 2 \probability{\state | \observation} \left( \max_{\action \in \actions} \qvalue^{*}(\state, \action) - \qvalue^{*}(\state, \action( \observation))\right)
\end{align}
which shows the desired result. Note that the last inequality is because \[\entropy{\mathds{1}_{\altstate}(\randomvar{\state}) | \randomvar{\observation} } \Delta(\altstate) \geq 0 = 2 \probability{\state | \observation} \left( \max_{\action \in \actions} \qvalue^{*}(\state, \action) - \qvalue^{*}(\state, \action( \observation))\right)\] since $\qvalue^{*}(\altstate, \action( \observation)) = \max_{\action \in \actions} \qvalue^{*}(\state, \action)$.

\textit{Case \ref{case:likelyloses}:}
Let $\action^{*} \in \max_{\action \in \actions} \qvalue^{*}(\altstate, \action)$. Since $\action^{*} \neq \action(\observation)$, due to the action selection rule, we know that 
\[\sum_{\state \in \states \setminus \lbrace \altstate \rbrace }\probability{\state | \observation} \qvalue^{*}(\state, \action(\observation)) + \probability{\altstate | \observation} \qvalue^{*}(\altstate, \action(\observation))  \geq  \sum_{\state \in \states \setminus \lbrace \altstate \rbrace }\probability{\state | \observation} \qvalue^{*}(\state, \action^{*}) + \probability{\altstate | \observation} \qvalue^{*}(\altstate, \action^{*})  \] which implies 
\[\sum_{\state \in \states \setminus \lbrace \altstate \rbrace }\probability{\state | \observation} \left (\qvalue^{*}(\state, \action(\observation)) - \qvalue^{*}(\state, \action^{*})\right) \geq \probability{\altstate | \observation} \left(\qvalue^{*}(\altstate, \action^{*}) - \qvalue^{*}(\altstate, \action(\observation)) \right) = \probability{\altstate | \observation} \left(\max_{\action \in \actions} \qvalue^{*}(\altstate, \action) - \qvalue^{*}(\altstate, \action(\observation)) \right)\] Noticing that $\left (\qvalue^{*}(\state, \action(\observation)) - \qvalue^{*}(\state, \action^{*})\right) \leq  \left( \max_{\action \in \actions} \qvalue^{*}(\state, \action) - \min_{\action \in \actions} \qvalue^{*}(\state, \action)\right) $, we also get 

\begin{equation}
    \sum_{\state \in \states \setminus \lbrace \altstate \rbrace }\probability{\state | \observation}  \left( \max_{\action \in \actions} \qvalue^{*}(\state, \action) - \min_{\action \in \actions} \qvalue^{*}(\state, \action)\right) \geq \probability{\altstate | \observation} \left(\max_{\action \in \actions} \qvalue^{*}(\altstate, \action) - \qvalue^{*}(\altstate, \action(\observation)) \right) \label{eq:valuelossdomination}
\end{equation}

Note that Case \ref{case:likelydominates} already shows \[      \sum_{\state \in \states} \entropy{\mathds{1}_{\state}(\randomvar{\state}) | \randomvar{\observation} } \Delta(\state)  \geq  \sum_{\observation} \probability{\observation} \left( \sum_{\state \in \states \setminus \lbrace \altstate \rbrace } 2 \probability{\state | \observation} \left( \max_{\action \in \actions} \qvalue^{*}(\state, \action) - \min_{\action \in \actions} \qvalue^{*}(\state, \action)\right)  \right)\] and \[  \sum_{\state \in \states} \entropy{\mathds{1}_{\state}(\randomvar{\state}) | \randomvar{\observation} } \Delta(\state)  \geq  \sum_{\observation} \probability{\observation} \left( \sum_{\state \in \states \setminus \lbrace \altstate \rbrace } 2 \probability{\state | \observation}\left( \max_{\action \in \actions} \qvalue^{*}(\state, \action) - \qvalue^{*}(\state, \action( \observation))\right)  \right)\] adding these inequalities and using \eqref{eq:valuelossdomination}, we get \[
    2\sum_{\state \in \states} \entropy{\mathds{1}_{\state}(\randomvar{\state}) | \randomvar{\observation} } \Delta(\state) \geq \sum_{\observation} \probability{\observation} \sum_{\state \in \states } 2 \probability{\state | \observation} \left( \max_{\action \in \actions} \qvalue^{*}(\state, \action) - \qvalue^{*}(\state, \action( \observation))\right)
\] which shows the desired result.
\end{proof}

\begin{lemma} \label{lemma:sim}
Let $\pi_{T}$ be a policy such that the agent follows the sensor selection rule Algorithm $1$ and the action selection rule \eqref{eq:qmdpactionselection} for $t$ time steps such that \(\sum_{\state \in \states}\entropy{\mathds{1}_{\state}(\randomvar{\state}_{t}) |  = \cup_{i \in \indexset} \randomvar{\observation}^{i} } \Delta(\state) \leq c\) for all $0\leq t \leq T$, then gets the actual state observations and follows $\pi^{*}$. Also, let $v_{T}$ be the expected return under $\pi_{T}$. Then \[v_{T-1} - v_{T} \leq \gamma^{T} c.\]
\end{lemma}

\begin{proof}
Let $I_{t}$ be the set of sensors choosen at time $t$.
    We have 
    \begin{align}
        v_{T} &=\mathbb{E}\left[\sum_{t=0}^{\infty} \discount^{t} \reward(\state_{t}, \action_{t}) \Bigg| \belief_{0}, \pi_{T} \right] 
\\
&= \mathbb{E}\left[\sum_{t=0}^{T} \discount^{t} \reward(\state_{t}, \action_{t}) \Bigg| \belief_{0}, \text{$I_{0}, \ldots, I_{T}$ satisfies Algorithm \eqref{algo:rationalinattention}, $\action_{0},\ldots,\action_{T}$ satisfies \eqref{eq:qmdpactionselection}}\right] 
\\
&\quad + \mathbb{E}\left[\sum_{t=T+1}^{\infty} \discount^{t} \reward(\state_{t}, \action_{t}) \Bigg| \belief_{0}, \text{$I_{0}, \ldots, I_{T}$ satisfies Algorithm \eqref{algo:rationalinattention}, $\action_{0},\ldots,\action_{T}$ satisfies \eqref{eq:qmdpactionselection}}, \action_{t} \sim \pi^{*}(\state_{t}) \text{ for all $t \geq T+1$}\right]
\\
&= \mathbb{E}\left[\sum_{t=0}^{T-1} \discount^{t} \reward(\state_{t}, \action_{t}) \Bigg| \belief_{0}, \text{$I_{0}, \ldots, I_{T-1}$ satisfies Algorithm \eqref{algo:rationalinattention}, $\action_{0},\ldots,\action_{T-1}$ satisfies \eqref{eq:qmdpactionselection}}\right]
\\
&\quad+\mathbb{E}\left[ \discount^{T} \reward(\state_{T}, \action_{T}) \Bigg| \belief_{0}, \text{$I_{0}, \ldots, I_{T}$ satisfies Algorithm \eqref{algo:rationalinattention}, $\action_{0},\ldots,\action_{T}$ satisfies \eqref{eq:qmdpactionselection}}\right] 
\\
&\quad + \mathbb{E}\left[\sum_{t=T+1}^{\infty} \discount^{t} \reward(\state_{t}, \action_{t}) \Bigg| \belief_{0}, \text{$I_{0}, \ldots, I_{T}$ satisfies Algorithm \eqref{algo:rationalinattention}, $\action_{0},\ldots,\action_{T}$ satisfies \eqref{eq:qmdpactionselection}}, \action_{t} \sim \pi^{*}(\state_{t}) \text{ for all $t \geq T+1$}\right].
    \end{align}
     Note that     \begin{align}
         &\mathbb{E}\left[\sum_{t=T+1}^{\infty} \discount^{t} \reward(\state_{t}, \action_{t}) \Bigg| \belief_{0}, \text{$I_{0}, \ldots, I_{T}$ satisfies Algorithm \eqref{algo:rationalinattention}, $\action_{0},\ldots,\action_{T}$ satisfies \eqref{eq:qmdpactionselection}}, \action_{t} \sim \pi^{*}(\state_{t}) \text{ for all $t \geq T+1$}\right]
         \\
         &=  \mathbb{E}\left[\gamma^{T+1} \gamevalue^{*}(s_{T+1}) \Bigg| \belief_{0}, \text{$I_{0}, \ldots, I_{T}$ satisfies Algorithm \eqref{algo:rationalinattention}, $\action_{0},\ldots,\action_{T}$ satisfies \eqref{eq:qmdpactionselection}}\right]
    \end{align} which implies
\begin{align}
        v_{T} 
&= \mathbb{E}\left[\sum_{t=0}^{T-1} \discount^{t} \reward(\state_{t}, \action_{t}) \Bigg| \belief_{0}, \text{$I_{0}, \ldots, I_{T-1}$ satisfies Algorithm \eqref{algo:rationalinattention}, $\action_{0},\ldots,\action_{T-1}$ satisfies \eqref{eq:qmdpactionselection}}\right]
\\
&\quad+\mathbb{E}\left[ \discount^{T} \reward(\state_{T}, \action_{T}) + \gamma \mathbb{E}\left[ \gamevalue^{*}(\state_{T+1}) \right] \Bigg| \belief_{0}, \text{$I_{0}, \ldots, I_{T}$ satisfies Algorithm \eqref{algo:rationalinattention}, $\action_{0},\ldots,\action_{T}$ satisfies \eqref{eq:qmdpactionselection}}\right]
\\
&= \mathbb{E}\left[\sum_{t=0}^{T-1} \discount^{t} \reward(\state_{t}, \action_{t}) \Bigg| \belief_{0}, \text{$I_{0}, \ldots, I_{T-1}$ satisfies Algorithm \eqref{algo:rationalinattention}, $\action_{0},\ldots,\action_{T-1}$ satisfies \eqref{eq:qmdpactionselection}}\right]
\\
&\quad+\mathbb{E}\left[ \discount^{T} \qvalue^{*}(s_{T}, \action_{T}) \Bigg| \belief_{0}, \text{$I_{0}, \ldots, I_{T}$ satisfies Algorithm \eqref{algo:rationalinattention}, $\action_{0},\ldots,\action_{T}$ satisfies \eqref{eq:qmdpactionselection}}\right].
    \end{align}

Also, note that 
\begin{align}
        v_{T-1} 
&= \mathbb{E}\left[\sum_{t=0}^{T-1} \discount^{t} \reward(\state_{t}, \action_{t}) \Bigg| \belief_{0}, \text{$I_{0}, \ldots, I_{T-1}$ satisfies Algorithm \eqref{algo:rationalinattention}, $\action_{0},\ldots,\action_{T-1}$ satisfies \eqref{eq:qmdpactionselection}}\right]
\\
&\quad+\mathbb{E}\left[ \discount^{T} \reward(\state_{T}, \action_{T}) + \gamma \mathbb{E}\left[ \gamevalue^{*}(\state_{T+1}) \right] \Bigg| \belief_{0}, \text{$I_{0}, \ldots, I_{T-1}$ satisfies Algorithm \eqref{algo:rationalinattention}, $\action_{0},\ldots,\action_{T-1}$ satisfies \eqref{eq:qmdpactionselection}, $\action_{T} \sim \pi^{*}(\state_{T})$} \right]
\\
&= \mathbb{E}\left[\sum_{t=0}^{T-1} \discount^{t} \reward(\state_{t}, \action_{t}) \Bigg| \belief_{0}, \text{$I_{0}, \ldots, I_{T-1}$ satisfies Algorithm \eqref{algo:rationalinattention}, $\action_{0},\ldots,\action_{T-1}$ satisfies \eqref{eq:qmdpactionselection}}\right]
\\
&\quad+\mathbb{E}\left[ \discount^{T} \max_{\action_{T} \in \actions} \qvalue^{*}(\state_{T}, \action_{T}) \Bigg| \belief_{0}, \text{$I_{0}, \ldots, I_{T-1}$ satisfies Algorithm \eqref{algo:rationalinattention}, $\action_{0},\ldots,\action_{T-1}$ satisfies \eqref{eq:qmdpactionselection}} \right].
    \end{align}

Taking the difference between $v_{T}$ and $v_{T-1}$, \begin{align}
        v_{T-1} - v_{T}
&= \mathbb{E}\left[ \discount^{T} \max_{\action \in \actions} \qvalue^{*}(\state_{T}, \action) \Bigg| \belief_{0}, \text{$I_{0}, \ldots, I_{T-1}$ satisfies Algorithm \eqref{algo:rationalinattention}, $\action_{0},\ldots,\action_{T-1}$ satisfies \eqref{eq:qmdpactionselection}} \right]
\\
&- \mathbb{E}\left[ \discount^{T} \qvalue^{*}(s_{T}, \action_{T}) \Bigg| \belief_{0}, \text{$I_{0}, \ldots, I_{T}$ satisfies Algorithm \eqref{algo:rationalinattention}, $\action_{0},\ldots,\action_{T}$ satisfies \eqref{eq:qmdpactionselection}}\right]
\\
&= \mathbb{E}\left[ \discount^{T} \left(\max_{\action \in \actions} \qvalue^{*}(\state_{T}, \action) - \qvalue^{*}(s_{T}, \action_{T}) \right) \Bigg| \belief_{0}, \text{$I_{0}, \ldots, I_{T-1}$ satisfies Algorithm \eqref{algo:rationalinattention}, $\action_{0},\ldots,\action_{T-1}$ satisfies \eqref{eq:qmdpactionselection}} \right].
    \end{align}

We have \(\sum_{\state \in \states}\entropy{\mathds{1}_{\state}(\randomvar{\state}_{T}) |  = \cup_{i \in \indexset} \randomvar{\observation}^{i} } \Delta(\state) \leq c\) and $\action_{T} \in \arg \max_{\action \in \actions} \sum_{\state \in \states} \belief'_{T}(\state) \qvalue^{*}(\state, \action)$. Combining these facts with Lemma \ref{lemma:stepvalueloss} and the above expression for $v_{T-1} - v_{T}$, $v_{T-1} - v_{T} \leq \gamma^{T} c$.
\end{proof}

\begin{proof}[Proposition \ref{prop:expectedvalueloss}]

Let $\pi_{T}$ be a sequence of policies defined as in Lemma \ref{lemma:sim} and $v_{T}$ be their respective expected returns. 

Due to Lemma \ref{lemma:sim}, we have
\begin{align}
    v_{-1} - v_{\infty} = \sum_{T = -1}^{\infty} (v_{T} - v_{T+1}) \leq  \sum_{T = -1}^{\infty} \gamma^{T+1} c = \frac{c}{1-\gamma}.
\end{align}

By noting that $v_{\infty} = v$ as defined in the proposition and $v_{-1} = \gamevalue^{*}(\initialstate) $, we get the desired result.
\end{proof}

\begin{proof}[Proposition \ref{prop:deceptiongain}]
The decision rule \eqref{eq:pl2distselection} assumes that both players play the subgame Nash equilibrium strategies in the preceding timesteps, and at the current time step, Player $2$ changes its action distribution from that of the equilibrium policy only if it improves its expected return over the security value given the action distribution of Player $1$.

Consider a scenario where the current time step is $0$ and the players act as described above. Then the expected return for player $2$ can only improve since player $1$ deviated from the equilibrium action distribution. 

Next, consider that the players act as described above for two time steps. Player $2$ made its decision at time $0$, considering that it will collect the security value for all branches.  But, applying the same logic to the value for player $2$ at time $1$, the expected return is better than the security value. Since the expected return improves for all branches, the actual expected return at time $0$ is also better than the security value that Player $2$ guaranteed at time $0$. 

Applying this idea recursively to all branches, we observe that the value at all nodes improve upon the security value and therefore $v \geq \gamevalue(\initialstate)$.
\end{proof}